\documentclass{birkjour}

\usepackage{tikz}
\usetikzlibrary{matrix,arrows,shapes}

\usepackage{amsmath,amssymb,amscd,mathrsfs}

 \newtheorem{thm}{Theorem}[section]

 \newtheorem{prop}[thm]{Proposition}
 \theoremstyle{definition}
 \newtheorem{definition}[thm]{Definition}
 \theoremstyle{remark}
 
 \newtheorem{example}[thm]{Example}
 \numberwithin{equation}{section}

\newcommand{\CoinX}[1]{C_0^\infty({#1})}

\newcommand{\Af}{{\mathscr{A}}}
\newcommand{\If}{{\mathscr{I}}}

\newcommand{\supp}{{\rm supp}\,}

\newcommand{\CC}{{\mathbb C}}
\newcommand{\RR}{{\mathbb R}}
\newcommand{\NN}{{\mathbb N}}

\newcommand{\ZZ}{{\mathbb Z}}

\newcommand{\Cb}{{\boldsymbol{C}}}
\newcommand{\Db}{{\boldsymbol{D}}}
\newcommand{\Ib}{{\boldsymbol{I}}}
\newcommand{\Lb}{{\boldsymbol{L}}}
\newcommand{\Mb}{{\boldsymbol{M}}}
\newcommand{\Nb}{{\boldsymbol{N}}}

\newcommand{\nto}{\stackrel{.}{\to}}
\newcommand{\nlto}{\stackrel{.}{\longrightarrow}}

\newcommand{\id}{{\rm id}}
\newcommand{\II}{{\mathbf{1}}}

\newcommand{\Ts}{{\sf T}}

\newcommand{\Funct}{{\sf Fun}}

\newcommand{\Loc}{{\sf Loc}}
\newcommand{\LCT}{{\sf LCT}}

\newcommand{\Sys}{{\sf Sys}}

\newcommand{\Alg}{{\sf Alg}}
\newcommand{\CAlg}{{\sf C^*\hbox{-}Alg}}

\newcommand{\Phys}{{\sf Phys}}
\newcommand{\Sympl}{{\sf Sympl}}

\newcommand{\Ac}{{\mathcal A}}
\newcommand{\Bc}{{\mathcal B}}
\newcommand{\KK}{{\mathscr K}}
\newcommand{\Bf}{{\mathscr B}}

\newcommand{\Lf}{{\mathscr L}}

\newcommand{\Tf}{{\mathscr T}}

\newcommand{\Aut}{{\rm Aut}}

\newcommand{\dvol}{d{\rm vol}}

\newcommand{\LL}{\mathcal{L}}

\newcommand{\hb}{\boldsymbol{h}}
\newcommand{\fb}{\boldsymbol{f}}

\newcommand{\rce}{\text{rce}}
\newcommand{\kin}{\text{kin}}
\newcommand{\dyn}{\text{dyn}}

\begin{document}

\title[The same physics in all spacetimes]{On the notion of `the same physics \\ in all spacetimes'}
\author[CJ Fewster]{Christopher J Fewster}
\address{Department of Mathematics, University of York, Heslington, York YO10 5DD, United Kingdom}
\email{chris.fewster@york.ac.uk}
\date{\today}

\begin{abstract}
Brunetti, Fredenhagen and Verch (BFV) have shown how the notion
of local covariance for quantum field theories can be formulated 
in terms of category theory: a theory being described as
a functor from a category of spacetimes to a category of $(C)^*$-algebras. We discuss whether this condition is sufficient to guarantee that a theory represents `the same physics' in all spacetimes, giving examples to show that it does not. 
A new criterion, {\em dynamical locality}, is formulated, which requires that descriptions of local physics based
on kinematical and dynamical considerations should coincide. 
Various applications are given, including a proof that dynamical locality for quantum fields is incompatible with the possibility of covariantly choosing a  preferred state in each spacetime. 

As part of this discussion we state a precise
condition that should hold on any class of theories each representing the 
same physics in all spacetimes. This condition holds for the dynamically local theories but is violated by the full class of locally covariant theories in the BFV sense.

The majority of results described form part of forthcoming papers with
Rainer Verch~\cite{FewVer:dynloc}.
\end{abstract}
\maketitle

\section{Introduction}

This contribution is devoted to the issue of how a physical theory should be formulated in arbitrary
spacetime backgrounds in such a way that the physical content is preserved. 
Our motivation arises from various directions. First, it is essential for the extension of axiomatic quantum field theory to curved spacetimes.
Second, one would expect that any quantization of gravity coupled to matter should, in certain
regimes, resemble a common theory of matter on different fixed backgrounds. Third, as our universe appears to
be well-described by a curved spacetime on many scales, one may well wonder what physical relevance should
be ascribed to a theory that could only be formulated in Minkowski space. 

In Lagrangian theories, there is an apparently satisfactory answer to our question: namely that the  Lagrangian should transform covariantly under coordinate
transformations. Actually, even here there are subtleties, as a simple example reveals. Consider
the nonminimally coupled scalar field with Lagrangian density
\[
\LL = \frac{1}{2}\sqrt{-g}\left(\nabla^a \phi\nabla_a\phi - \xi R \phi^2\right).
\]
In Minkowski space the equation of motion does not depend on the coupling constant $\xi$ but theories
with different coupling constant can still be distinguished by their stress-tensors, which contain terms proportional to 
$\xi$. Suppose, however, that $\xi R$ is replaced by $\zeta(R)$, where $\zeta$ is a smooth function vanishing in a neighbourhood
of the origin and taking a constant value $\xi_0\neq 0$ for all sufficiently large $|R|$. 
This gives a new theory that coincides with the $\xi=0$ theory in Minkowski space in 
terms of the field equation, stress-energy tensor and any other quantity formed by functional differentiation
of the action and then evaluated in Minkowski space. 
But in de Sitter space (with sufficiently large cosmological constant) the theory coincides, in the same sense, 
with the $\xi=\xi_0$ theory. Of course, it is unlikely that a theory of this type would be physically relevant, but the example serves to illustrate that covariance of the Lagrangian does not guarantee that the physical content of a 
theory will be the same in all spacetimes. Additional conditions would be required;
perhaps that the (undensitized) Lagrangian should depend analytically on the metric and curvature quantities. 

The situation is more acute when one attempts to generalize axiomatic quantum field theory to the
curved spacetime context. After all, these approaches do not take a classical action as their starting point,
but focus instead on properties that ought to hold for any reasonable quantum field theory. In Minkowski space, Poincar\'e covariance
and the existence of a unique invariant vacuum state obeying the spectrum condition provide strong constraints which ultimately account to a large part for the successes of axiomatic
QFT in both its Wightman--G{\aa}rding~\cite{StreaterWightman} and Araki--Haag--Kastler~\cite{Haag} formulations. As a generic spacetime
has no symmetries, and moreover attempts to define distinguished vacuum states in 
general curved spacetimes lead to failure even for free fields,\footnote{Theorem~\ref{thm:natural_states} below strengthens this to a general no-go result.} there are severe difficulties associated with 
generalizing the axiomatic setting to curved spacetime.

Significant progress was made by Brunetti, Fredenhagen and Verch \cite{BrFrVe03} (henceforth BFV) who formalized the
idea of a locally covariant quantum field theory in the language of category theory (see~\cite{MacLane,AdamekHerrlichStrecker} as general references). As we will see, 
this paper opens up new ways to analyze physical theories; in this sense it justifies its
subtitle, `A new paradigm for local quantum physics'. At the structural level, 
the ideas involved have led to a number of new model-independent results for QFT 
in curved spacetime such as the spin-statistics connection~\cite{Verch01}, 
Reeh--Schlieder type results~\cite{Sanders_ReehSchlieder}
and the analysis superselection sectors~\cite{Br&Ru05,BrunettiRuzzi_topsect}; they were also crucial in completing the perturbative construction
of interacting QFT in curved spacetime \cite{BrFr2000,Ho&Wa01,Ho&Wa02}. It should also be mentioned that antecedants of the ideas underlying the BFV formalism can be found in~\cite{Fulling73,Kay79,Dimock1980}. However, we stress that the BFV approach is not
simply a matter of formalism: it suggests and facilitates new calculations that can lead to
concrete physical predictions such as {\em a priori} bounds on Casimir energy densities~\cite{Few&Pfen06, Fewster2007} and new viewpoints in cosmology~\cite{DapFrePin2008,DegVer2010} 
(see also Verch's contribution to these proceedings~\cite{VerchRegensburg}). 

The focus in this paper is on the physical content of local covariance in the BFV
formulation: in particular, is it sufficiently strong to enforce the same physics in all spacetimes?
And, if not, what does? We confine ourselves here to the main ideas, referring to 
forthcoming papers~\cite{FewVer:dynloc} for the details.

\section{Locally covariant physical theories}

To begin, let us consider what is required for a local, causal description of physics. A minimal list might 
be the following:
\begin{itemize}
\item Experiments can be conducted over a finite timespan and spatial extent in reasonable isolation from the rest of the world.
\item We may distinguish a `before' and an `after'  for such experiments. 
\item The `same' experiment could, in principle, be conducted at other locations and times with the `same' outcomes (to within experimental accuracy, and possibly in a statistical sense). 
\item The theoretical account of such experiments should be independent (as far as possible) of the rest of the world. 
\end{itemize}
Stated simply, we should not need to know where in the Universe our laboratory is, nor whether the
Universe has any extent much beyond its walls (depending on the duration of the experiment
and the degree of shielding the walls provide). Of course, these statements contain a number of undefined terms; and, when referring to the `same' experiment, we should take into account the motion of the apparatus relative to local inertial frames 
(compare a Foucault pendulum in Regensburg (latitude $49^\circ$N) with one in York ($54^\circ$N)). Anticipating Machian objections to the last requirement, we assume that sufficient structures are present to permit identification of (approximately) inertial frames.

If spacetime is modelled as a Lorentzian manifold, the minimal requirements can be met by restricting
to the globally hyperbolic spacetimes.
Recall that an orientable and time-orientable Lorentzian manifold $\Mb$ (the symbol $\Mb$ will encompass the manifold, choice of metric,\footnote{We adopt signature $+-\cdots-$.}
orientation and time-orientation) is said to be globally hyperbolic if
it has no closed causal curves and $J_\Mb^+(p)\cap J_\Mb^-(q)$ is compact for all $p,q\in \Mb$ (see~\cite[Thm 3.2]{Bernal:2006xf} for
equivalence with the older definition~\cite{HawkingEllis}).\footnote{Here, $J^\pm_\Mb(p)$ 
denotes the set of points that can be reached by future ($+$) or past ($-$) directed causal curves
originating from a point $p$ in $\Mb$ (including $p$ itself). If $S$ is a subset of $\Mb$, we write $J^\pm(S)$ for the union $J^\pm(S)=\bigcup_{p\in S} J^\pm_\Mb(p)$, and $J_\Mb(S)=J_\Mb^+(S)\cup
J_\Mb^-(S)$.} Compactness of these regions
ensures that experiments can be conducted within bounded spacetime regions,
which could in principle be shielded, and that communication within the region can be achieved by
finitely many radio links of finite range. Technically, of course, the main use of global hyperbolicity is 
that it guarantees well-posedness of the Klein--Gordon equation and other field equations with metric principal part (as described, for example, in Prof.~B\"ar's contribution to this workshop).
Henceforth we will assume that all spacetimes are globally hyperbolic and have at most finitely many connected components.

The globally hyperbolic spacetimes constitute the objects of a category $\Loc$, in 
which the morphisms are taken to be hyperbolic embeddings.
\begin{definition} 
A \emph{hyperbolic embedding} of $\Mb$ into $\Nb$ is an
isometry $\psi:\Mb\to\Nb$ preserving time and space orientations and such that 
$\psi(\Mb)$ is a causally convex\footnote{A subset $S$ of $\Nb$ is
causally convex if every causal curve in $\Nb$ with endpoints in $S$ is contained wholly in $S$.} (and hence globally hyperbolic) subset of $\Nb$.
\end{definition}

A hyperbolic embedding of $\psi:\Mb\to\Nb$ allows us to regard $\Nb$ as
an enlarged version of $\Mb$: any experiment taking place in $\Mb$
should have an analogue in $\psi(\Mb)$ which should yield
indistinguishable results -- at least if `physics is the same' in both spacetimes.

BFV implemented this idea by regarding a physical theory as a functor from
the category $\Loc$ to a category $\Phys$, which encodes the type of physical system under
consideration: the objects representing systems and the morphisms representing embeddings
of one system in a larger one. BFV were
interested in quantum field theories, described in terms of their algebras of
observables. Here, natural candidates for $\Phys$ are provided by $\Alg$ (resp., $\CAlg$), 
whose objects are unital $(C)^*$-algebras with unit-preserving injective $*$-homomorphisms
as the morphisms. However, the idea applies more widely. For example, in the context of 
linear Hamiltonian systems, a natural choice of $\Phys$ would be
the category $\Sympl$ of real symplectic spaces with symplectic maps as morphisms. As a more
elaborate example, take as 
objects of a category $\Sys$ all pairs $(\Ac,S)$ where $\Ac\in\Alg$ (or $\CAlg$) and $S$ is a nonempty
convex subset of the states on $\Ac$, closed under operations induced by 
$\Ac$.\footnote{That is, if $\sigma\in S$ and $A\in\Ac$ has $\sigma(A^*A)=1$, then
$\sigma_A(X):=\sigma(A^*X A)$ defines $\sigma_A\in S$.} Whenever $\alpha:\Ac\to\Bc$ is an $\Alg$-morphism such that
$\alpha^*(T)\subset S$, we will say that $\alpha$ induces a morphism $(\Ac,S)\to (\Bc,T)$ in $\Sys$.
It may be shown that $\Sys$ is a category under the composition inherited from $\Alg$. This category
provides an arena for discussing algebras of observables equipped with a distinguished class of states. 
The beauty of the categorical description is that it allows us to treat different types of physical
theory in very similar ways. 

According to BFV, then, a theory should assign to each spacetime $\Mb\in\Loc$ a mathematical object  $\Af(\Mb)$ modelling `the physics on $\Mb$'  and to each hyperbolic embedding $\psi:\Mb\to\Nb$ 
a means of embedding the physics on $\Mb$ inside the physics on $\Nb$, expressed as
a morphism $\Af(\psi):\Af(\Mb)\to\Af(\Nb)$. The natural conditions that 
$\Af(\id_\Mb)=\id_{\Af(\Mb)}$ and $\Af(\varphi\circ\psi) = \Af(\varphi)\circ\Af(\psi)$ are
precisely those that make $\Af$ a (covariant) functor. 

In particular, this gives an immediate definition of the `local physics' in any nonempty open
globally hyperbolic subset $O$ of $\Mb$ as
\[
\Af^\kin(\Mb;O) := \Af(\Mb|_O), 
\]
where $\Mb|_O$ denotes the region $O$ with geometry induced from $\Mb$
and considered as a spacetime in its own right; this embeds in $\Af(\Mb)$ by
\[
\Af^\kin(\Mb;O) \xrightarrow{\Af(\iota_{\Mb;O})} \Af(\Mb),
\]
where $
\iota_{\Mb;O}:\Mb|_O\to \Mb$ is the obvious embedding. 
Anticipating future developments, we refer to this as a {\em kinematic} description of
the local physics. 

The kinematic description has many nice properties. Restricted to Mink{\-}owski space, and with $\Phys=\CAlg$, BFV were able to show that the net $O\mapsto \Af^\kin(\Mb;O)$ 
(for nonempty open, relatively compact, globally hyperbolic $O$) satisfies the Haag--Kastler 
axioms of algebraic QFT.\footnote{More precisely, BFV study the images 
$\Af(\iota_{\Mb;O})(\Af^\kin(\Mb;O))$ as sub-$C^*$-algebras of $\Af(\Mb)$.} 
In general spacetimes, it therefore provides a generalisation of the Haag--Kastler framework. 
It is worth focussing on one particular issue. Dimock~\cite{Dimock1980} also articulated a 
version of Haag--Kastler axioms for curved spacetime QFT and also expressed this partly in functorial language. However, Dimock's covariance axiom was global in nature: it required that when
spacetimes $\Mb$ and $\Nb$ are isometric then there should be an isomorphism between the 
nets of algebras on $\Mb$ and $\Nb$. In the BFV framework this idea is localised and 
extended to situations in which $\Mb$ is hyperbolically embedded in $\Nb$ but not necessarily 
globally isometric to it. To be specific, suppose that there is a morphism $\psi:\Mb\to\Nb$
and that $O$ is a nonempty, open globally hyperbolic subset of $\Mb$ with finitely many 
connected components. Then $\psi(O)$ also obeys these conditions in $\Nb$ and, moreover,
restricting the domain of $\psi$ to $O$ and the codomain to $\psi(O)$, we obtain an 
isomorphism $\hat{\psi}:\Mb|_O\to\Nb|_{\psi(O)}$ making the diagram
\[
\begin{tikzpicture}[baseline=0 em,description/.style={fill=white,inner sep=2pt}]
\matrix (m) [ampersand replacement=\&,matrix of math nodes, row sep=3em,
column sep=2.5em, text height=1.5ex, text depth=0.25ex]
{ \Mb|_O \&  \Nb|_{\psi(O)} \\
 \Mb \& \Nb\\ };
\path[->,font=\scriptsize]
(m-1-1) edge node[left] {$\iota_{\Mb;O}$} (m-2-1)
(m-1-2) edge node[auto] {$\iota_{\Nb;\psi(O)}$} (m-2-2)
(m-2-1) edge node[auto] {$\psi$} (m-2-2);
%(m-2-2) edge node[below] {$ \varphi(\psi)_\Nb $} (m-2-3)
%(m-1-2) edge[color=red,ultra thick] node[above,sloped,color=red] {$ \varphi_\Delta(\psi) $} (m-2-3);
\path[->,font=\scriptsize]
(m-1-1) edge node[above]{$\hat{\psi}$}  node[below] {$\cong$} (m-1-2);
\end{tikzpicture}
\]
commute. As functors always map commuting diagrams to commuting diagrams, and isomorphisms to isomorphisms, we obtain the commuting diagram
\begin{equation}
\begin{tikzpicture}[baseline=0 em,description/.style={fill=white,inner sep=2pt}]
\matrix (m) [ampersand replacement=\&,matrix of math nodes, row sep=3em,
column sep=2.5em, text height=1.5ex, text depth=0.25ex]
{ \Af^\kin(\Mb;O) \&  \Af^\kin(\Nb;\psi(O)) \\
 \Af(\Mb) \&  \Af(\Nb)\\ };
\path[->,font=\scriptsize]
(m-1-1) edge node[left] {$\alpha^\kin_{\Mb;O}$} (m-2-1)
(m-1-2) edge node[auto] {$\alpha^\kin_{\Nb;\psi(O)}$} (m-2-2)
(m-2-1) edge node[auto] {$\Af(\psi)$} (m-2-2);
%(m-2-2) edge node[below] {$ \varphi(\psi)_\Nb $} (m-2-3)
%(m-1-2) edge[color=red,ultra thick] node[above,sloped,color=red] {$ \varphi_\Delta(\psi) $} (m-2-3);
\path[->,font=\scriptsize]
(m-1-1) edge node[above] {$\Af(\hat{\psi})$} node[below]  {$\cong$} (m-1-2);
\end{tikzpicture}\label{eq:kinematic_covariance_diagram}
\end{equation}
in which, again anticipating future developments, we have written $\alpha^\kin_{\Mb;O}$ for the
morphism $\Af(\iota_{\Mb;O})$ embedding $\Af^\kin(\Mb;O)$ in $\Af(\Mb)$. The significance of this
diagram is that it shows that the kinematic description of local physics is truly local: it assigns
equivalent descriptions of the physics to isometric subregions of the ambient spacetimes $\Mb$ and $\Nb$. 
This holds even when there is no hyperbolic embedding of one of these ambient spacetimes in the other, as can be seen if we consider a further morphism $\varphi:\Mb\to\Lb$, thus obtaining an isomorphism
$\Af(\hat{\psi})\circ \Af(\hat{\varphi}^{-1}): \Af^\kin(\Lb;\varphi(O))\to \Af^\kin(\Nb;\psi(O))$
even though there need be no morphism between $\Lb$ and $\Nb$.

\paragraph{Example: the real scalar field}
Take $\Phys=\Alg$.
The quantization of the Klein--Gordon equation $(\Box_\Mb+m^2)\phi=0$ is well understood in arbitrary
globally hyperbolic spacetimes. To each spacetime $\Mb\in\Loc$ we assign $\Af(\Mb)\in\Alg$
with generators $\Phi_\Mb(f)$, labelled by test functions $f\in\CoinX{\Mb}$ and interpreted as
smeared fields, subject to the following relations (which hold for arbitrary $f,f'\in\CoinX{\Mb}$):
\begin{itemize}
\item $f\mapsto\Phi_\Mb(f)$ is complex linear
\item $\Phi_\Mb(f)^*=\Phi_\Mb(\overline{f})$
\item $\Phi_\Mb((\Box_\Mb+m^2)f) = 0$
\item $[\Phi_\Mb(f),\Phi_\Mb(f')] = iE_\Mb(f,f')\II_{\Af(\Mb)}$.
\end{itemize}
Here, $E_\Mb$ is the advanced-minus-retarded Green function
for $\Box_\Mb+m^2$, 
whose existence is guaranteed by global hyperbolicity of $\Mb$.
Now if  $\psi:\Mb\to\Nb$ is a hyperbolic embedding 
we have
\begin{equation}\label{eq:intertwine}
\psi_* \Box_\Mb f = \Box_\Nb \psi_* f \quad\textrm{and}\quad E_\Nb(\psi_* f,\psi_* f') = E_\Mb(f,f')
\end{equation}
for all $f,f'\in\CoinX{\Mb}$,
where 
\[
(\psi_* f)(p) =\begin{cases} f(\psi^{-1}(p)) & p\in\psi(\Mb)\\
0 & \text{otherwise.}
\end{cases}
\]
The second assertion in~\eqref{eq:intertwine} follows from the first, together with the uniqueness of advanced/retarded solutions to the
inhomogeneous Klein--Gordon equation. In consequence, the map
\[
\Af(\psi)(\Phi_\Mb(f)) = \Phi_\Nb(\psi_* f), \quad \Af(\psi)\II_{\Af(\Mb)} = \II_{\Af(\Nb)}
\]
extends to a $*$-algebra homomorphism $\Af(\psi):\Af(\Mb)\to\Af(\Nb)$. Furthermore, 
this is a monomorphism because $\Af(\Mb)$ is simple, so $\Af(\psi)$ is indeed a morphism
in $\Alg$. It is clear that $\Af(\psi\circ\varphi)=\Af(\psi)\circ\Af(\varphi)$ and $\Af(\id_\Mb)=\id_{\Af(\Mb)}$ because these equations hold on the generators, so
$\Af$ indeed defines a functor from $\Loc$ to $\Alg$. Finally, if $O$ is a nonempty open
globally hyperbolic subset of $\Mb$, $\Af^\kin(\Mb;O)$
is defined as $\Af(\Mb|_O)$; its image under $\Af(\iota_{\Mb;O})$
may be characterized as the unital subalgebra of $\Af(\Mb)$ generated by those
$\Phi_\Mb(f)$ with $f\in\CoinX{O}$.

\section{The time-slice axiom and relative Cauchy evolution}

The structures described so far may be regarded as kinematic.
To introduce dynamics we need a replacement for the idea that
evolution is determined by data on a Cauchy surface. With this
in mind, we say that a hyperbolic embedding $\psi:\Mb\to\Nb$ is \emph{Cauchy} if $\psi(\Mb)$ contains a
Cauchy surface of $\Nb$ and say that a locally covariant theory $\Af$ satisfies the \emph{time-slice axiom}
if $\Af(\psi)$ is an isomorphism whenever $\psi$ is Cauchy. 
%\begin{figure}
%\begin{center}
%\begin{tikzpicture}
%\draw[fill=lightgray] (4,1.25) -- ++(2,0) -- ++(0,2) -- ++(-2,0) -- cycle;
%\draw[fill=green] (4,2) -- ++(2,0) -- ++(0,0.5) -- ++(-2,0) -- cycle;
%\draw[fill=green] (0,2) -- ++(2,0) -- ++(0,0.5) -- ++(-2,0) -- cycle;
%\draw[color=red,thick] (4,2.25) -- ++(2,0);
%\draw[color=blue,line width=4pt,->] (2.25,2.25) -- (3.75,2.25) node[pos=0.5,above]{$\psi$};
%\node[anchor=north] at (5,1.25) {$\Nb$};
%\node[anchor=north] at (1,2) {$\Mb$};
%\end{tikzpicture}
%\end{center}
%\caption{A schematic representation of a Cauchy morphism}
%\end{figure}
The power of the time-slice axiom arises as follows. Any Cauchy surface naturally inherits an orientation from the ambient spacetime; if two spacetimes $\Mb$ and $\Nb$ have Cauchy surfaces that are equivalent modulo an
orientation-preserving diffeomorphism (not necessarily an isometry) then the two spacetimes may be linked by a chain of Cauchy
morphisms \cite{FullingNarcowichWald} (see Fig.~\ref{fig:deform}). If the time-slice axiom is satisfied,
then the theory $\Af$ assigns isomorphisms to each Cauchy morphism; this can often be used to infer that the theory on $\Nb$ obeys some property, given that it holds in $\Mb$.  
\begin{figure}
\begin{center}
\begin{tikzpicture}[scale=0.6]
\definecolor{orange}{rgb}{1,0.5,0}
\draw[fill=lightgray] (-7,2) +(-1,0) -- +(0,-1) -- +(1,0) -- +(0,1) -- cycle;
\draw[fill=lightgray] (-3,2) +(-1,0) -- +(0,1) -- +(1,0) ..  controls +(0,0.5) .. +(-1,0);
\draw[fill=lightgray] (1,2) +(-1,0) -- +(0,1) -- +(1,0) ..  controls +(0,0.5) .. +(-1,0);
\draw[fill=black] (1,2) +(-1,0) -- +(0,-1) -- +(1,0) ..  controls +(0,-0.5) .. +(-1,0);
\shadedraw[top color =lightgray,bottom color=darkgray] (1,2) +(-1,0) ..  controls +(0,0.5) .. +(1,0) ..  controls +(0,-0.5) .. +(-1,0);
\draw[fill=black] (5,2) +(-1,0) -- +(0,-1) -- +(1,0) ..  controls +(0,-0.5) .. +(-1,0);
\draw[fill=black] (9,2) +(-1,0) -- +(0,-1) -- +(1,0) -- +(0,1) -- cycle;
\draw[color=gray,line width=4pt,->] (-4,2.5) -- (-6,2.5);
\draw[color=gray,line width=4pt,<-] (0,2.5) -- (-2,2.5);
\draw[color=gray,line width=4pt,->] (4,1.5) -- (2,1.5);
\draw[color=gray,line width=4pt,->] (6,1.5) -- (8,1.5);
\node[anchor=north] at (-7,1) {$\Mb$};
\node[anchor=north] at (1,1) {$\Ib$};
\node[anchor=north] at (9,1) {$\Nb$};
\end{tikzpicture}
\end{center}
\caption{Schematic representation of the deformation construction in~\cite{FullingNarcowichWald}:
globally hyperbolic spacetimes $\Mb$ and $\Nb$, whose Cauchy surfaces are related by 
an orientation preserving diffeomorphism, can be 
linked by a chain of Cauchy morphisms and an interpolating spacetime $\Ib$.}\label{fig:deform}
\end{figure}
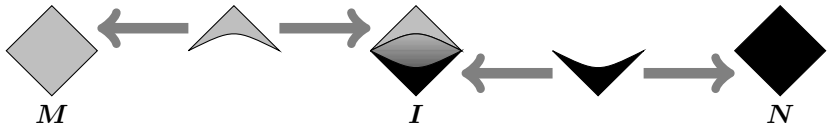

One of the innovative features of the BFV framework is its ability to
quantify the response of the theory to a perturbation in the metric. 
We write $H(\Mb)$ for the set of compactly supported smooth metric
perturbations $\hb$ on $\Mb$ that are sufficiently mild so as to modify
$\Mb$ to another globally hyperbolic spacetime $\Mb[\hb]$. In the 
schematic diagram of Fig.~\ref{fig:rce}, the metric perturbation lies
between two Cauchy surfaces. By taking a globally hyperbolic neighbourhood
of each Cauchy surface we are able to find Cauchy morphisms  $\iota^\pm$ and
$\iota^\pm[\hb]$ as shown.
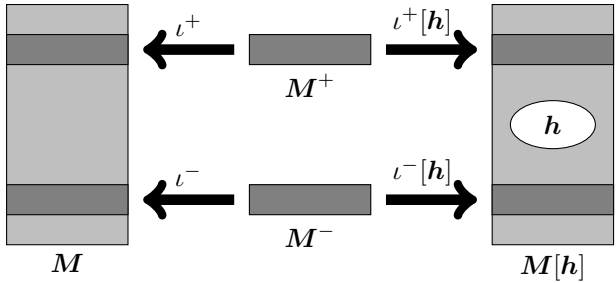
\begin{figure}
\begin{center}
\begin{tikzpicture}[scale=0.8]
\definecolor{Gold}{rgb}{.93,.82,.24}
\definecolor{Orange}{rgb}{1,0.5,0}
\draw[fill=lightgray] (-4,0) -- ++(2,0) -- ++(0,4) -- ++(-2,0) -- cycle;
\draw[fill=lightgray] (4,0) -- ++(2,0) -- ++(0,4) -- ++(-2,0) -- cycle;
\draw[fill=gray] (4,3) -- ++(2,0) -- ++(0,0.5) -- ++(-2,0) -- cycle;
\draw[fill=gray] (0,3) -- ++(2,0) -- ++(0,0.5) -- ++(-2,0) -- cycle;
\draw[fill=gray] (-4,3) -- ++(2,0) -- ++(0,0.5) -- ++(-2,0) -- cycle;
\draw[fill=gray] (4,0.5) -- ++(2,0) -- ++(0,0.5) -- ++(-2,0) -- cycle;
\draw[fill=gray] (0,0.5) -- ++(2,0) -- ++(0,0.5) -- ++(-2,0) -- cycle;
\draw[fill=gray] (-4,0.5) -- ++(2,0) -- ++(0,0.5) -- ++(-2,0) -- cycle;
\draw[color=black,line width=4pt,->] (2.25,3.25) -- (3.75,3.25) node[pos=0.4,above]{$\iota^+[\hb]$};
\draw[color=black,line width=4pt,->] (2.25,0.75) -- (3.75,0.75) node[pos=0.4,above]{$\iota^-[\hb]$};
\draw[color=black,line width=4pt,->] (-0.25,3.25) -- (-1.75,3.25) node[pos=0.5,above]{$\iota^+$};
\draw[color=black,line width=4pt,->] (-0.25,0.75) -- (-1.75,0.75) node[pos=0.5,above]{$\iota^-$};
\draw[fill=white] (5,2) ellipse (0.7 and 0.4);
\node at (5,2) {$\hb$};
\node[anchor=north] at (5,0) {$\Mb[\hb]$};
\node[anchor=north] at (-3,0) {$\Mb$};
\node[anchor=north] at (1,3) {$\Mb^+$};
\node[anchor=north] at (1,0.5) {$\Mb^-$};
\end{tikzpicture}
\end{center}
\caption{Morphisms involved in relative Cauchy evolution}\label{fig:rce}
\end{figure}
Now if $\Af$ obeys the time-slice axiom, any metric perturbation $\hb\in H(\Mb)$ defines an
automorphism of $\Af(\Mb)$,
\[
\rce_\Mb[\hb] =
\Af(\iota^-)\circ\Af(\iota^-[\hb])^{-1}\circ\Af(\iota^+[\hb])\circ\Af(\iota^+)^{-1}.
\]
BFV showed that the functional derivative of the relative Cauchy evolution defines 
a derivation on the algebra of observables which can be interpreted as a commutator
with a stress-energy tensor so that
\[
[\Ts_\Mb(\fb),A] =2i \left.\frac{d}{ds}\rce_\Mb[\hb(s)] A\right|_{s=0},
\]
where $\hb(s)$ is a smooth one-parameter family of metric perturbations in $H(\Mb)$
and $\fb=\dot{\hb}(0)$.  
In support of the interpretation of $\Ts$ as a stress-energy tensor we may cite the
fact that $\Ts$ is symmetric and conserved. Moreover, the stress-energy tensor 
in concrete models can be shown to satisfy the above relation (see BFV and~\cite{Sanders_dirac:2010} for the scalar and Dirac fields respectively) and we will see further evidence for this link below in Sect.~\ref{sect:KG}.

\section{The same physics in all spacetimes}
\subsection{The meaning of SPASs}

We can now begin to analyse the question posed at the start of this paper. 
However, is not yet clear what should be understood by saying that a theory
represents the same physics in all spacetimes (SPASs). Rather than give a direct answer
we instead posit a property that should be valid for any reasonable notion of
SPASs: If {\em two} theories individually represent the
same physics in all spacetimes, and there is {\em some} spacetime in which 
the theories coincide then they should coincide in {\em all} spacetimes.
In particular, we obtain the following necessary condition:

\begin{definition}
A class of theories $\mathfrak T$ has the SPASs property if no proper
subtheory $\Tf'$ of $\Tf$ in $\mathfrak T$ can fully account for the physics
of $\Tf$ in any single spacetime.
\end{definition}

This can be made precise by introducing a new category: the category of
locally covariant theories, $\LCT$, whose objects are functors from $\Loc$ to $\Phys$ 
and in which the morphisms are natural transformations between such functors.
Recall that a natural transformation $\zeta:\Af\nto\Bf$ between theories
$\Af$ and $\Bf$ assigns to each $\Mb$ a morphism $\Af(\Mb)\stackrel{\zeta_\Mb}{\longrightarrow}
\Bf(\Mb)$ so that for each hyperbolic embedding $\psi$, the following diagram commutes:
\[
 \begin{tikzpicture}[baseline=0 em, description/.style={fill=white,inner sep=2pt}]
\matrix (m) [ampersand replacement=\&,matrix of math nodes, row sep=3em,
column sep=2.5em, text height=1.5ex, text depth=0.25ex]
{\Mb \& \Af(\Mb) \&  \Bf(\Mb) \\
\Nb \& \Af(\Nb) \&  \Bf(\Nb)\\ };
\path[->,font=\scriptsize]
(m-1-1) edge node[left] {$\psi$} (m-2-1)
(m-1-2) edge node[auto] {$ \zeta_\Mb $} (m-1-3)
        edge node[auto,left] {$ \Af(\psi) $} (m-2-2)
(m-2-2) edge node[auto] {$ \zeta_\Nb $} (m-2-3)
(m-1-3) edge node[auto] {$ \Bf(\psi) $} (m-2-3);
\end{tikzpicture}
\]
that is, $\zeta_\Nb\circ\Af(\psi) = \Bf(\psi)\circ\zeta _\Mb$.
The interpretation is that $\zeta$ embeds $\Af$ as a sub-theory of $\Bf$.
If every $\zeta_\Mb$ is an isomorphism, $\zeta$ is an \emph{equivalence} of $\Af$ and $\Bf$;
if some $\zeta_\Mb$ is an isomorphism we will say that $\zeta$ is a \emph{partial equivalence}. 

\begin{example}\label{ex:copowers}
With $\Phys=\Alg$: given any $\Af\in\LCT$, define $\Af^{\otimes k}$ ($k\in \NN$) by
\[
\Af^{\otimes k}(\Mb) = \Af(\Mb)^{\otimes k},\qquad \Af^{\otimes k}(\psi)
= \Af(\psi)^{\otimes k}
\]
i.e., $k$ independent copies of $\Af$, where the algebraic tensor product is used. 
Then 
\begin{align*}
\eta_\Mb^{k,l}: \Af^{\otimes k}(\Mb) & \to \Af^{\otimes l}(\Mb)\\
A &\mapsto A\otimes \II_{\Af(\Mb)}^{\otimes (l-k)}
\end{align*}
defines a natural $\eta^{k,l}:\Af^{\otimes k}\nto\Af^{\otimes l}$ for
$k\le l$ and we have the obvious relations that 
$\eta^{k,m}=\eta^{l,m}\circ \eta^{k,l}$ if $k\le l\le m$. 
\end{example}

We can now formulate the issue of SPASs precisely:
\begin{definition} 
A class $\mathfrak{T}$ of theories in $\LCT$ has the SPASs property if
all partial equivalences between theories in $\mathfrak{T}$ are equivalences. 
\end{definition}
That is, if $\Af$ and $\Bf$ are theories in $\mathfrak{T}$, such that
$\Af$ is a subtheory of $\Bf$ with $\zeta:\Af\nto \Bf$ in $\LCT$,
and  $\zeta_\Mb:\Af(\Mb)\to\Bf(\Mb)$ is an isomorphism for some $\Mb$, then $\zeta$ is an equivalence.

\subsection{Failure of SPASs in $\LCT$}\label{sect:failure_of_SPASs}

Perhaps surprisingly, theories in $\LCT$ need not represent the same physics in all
spacetimes. Indeed, a large class of pathological theories may
be constructed as follows. Let us take any functor $\varphi:\Loc\to\LCT$, i.e.,
a locally covariant choice of locally covariant theory [we will give an example below]. 
Thus for each $\Mb$, $\varphi(\Mb)\in\LCT$ is a choice of theory defined in all spacetimes, 
and each hyperbolic embedding $\psi:\Mb\to\Nb$ corresponds to an embedding
$\varphi(\psi)$ of $\varphi(\Mb)$ as a sub-theory of $\varphi(\Nb)$. 

Given $\varphi$ as input, we may define a \emph{diagonal theory}
$\varphi_\Delta\in\LCT$ by setting
\[
\varphi_\Delta(\Mb) = \varphi(\Mb)(\Mb)\qquad
\varphi_\Delta(\psi) = \varphi(\psi)_{\Nb}\circ\varphi(\Mb)(\psi)
%=\varphi(\Nb)(\psi)\circ\varphi(\psi)_{\Mb}.
\]
for each spacetime $\Mb$ and hyperbolic embedding $\psi:\Mb\to\Nb$. The definition
of $\varphi_\Delta(\psi)$ is made more comprehensible by realising that it is
the diagonal of a natural square
\begin{center}
\begin{tikzpicture}[description/.style={fill=white,inner sep=2pt}]
\matrix (m) [ampersand replacement=\&,matrix of math nodes, row sep=3em,
column sep=2.5em, text height=1.5ex, text depth=0.25ex]
{ \Mb \& \varphi(\Mb)(\Mb) \&  \varphi(\Nb)(\Mb) \\
\Nb \&   \varphi(\Mb)(\Nb) \&  \varphi(\Nb)(\Nb)\\ };
\path[->,font=\scriptsize]
(m-1-1) edge node[auto] {$ \psi $} (m-2-1)
(m-1-2) edge node[auto] {$ \varphi(\psi)_\Mb $} (m-1-3)
        edge node[left] {$ \varphi(\Mb)(\psi) $} (m-2-2)
(m-1-3) edge node[auto] {$ \varphi(\Nb)(\psi) $} (m-2-3)
(m-2-2) edge node[below] {$ \varphi(\psi)_\Nb $} (m-2-3)
(m-1-2) edge[ultra thick] node[above,sloped] {$ \varphi_\Delta(\psi) $} (m-2-3);
\end{tikzpicture}
\end{center}
arising from the subtheory embedding $\varphi(\psi):\varphi(\Mb)\nto
\varphi(\Nb)$. 
The key point is that, while $\varphi_\Delta$ is easily shown to be a functor and therefore defines a locally covariant theory, there seems no reason to believe that $\varphi_\Delta$ should represent the same physics in all spacetimes. Nonetheless, $\varphi_\Delta$ may have many reasonable properties. For example, it 
satisfies the time-slice axiom if each $\varphi(\Mb)$ does and in addition $\varphi$ does (i.e., 
$\varphi(\psi)$ is an equivalence for every Cauchy morphism $\psi$). 

To demonstrate the existence of nontrivial functors $\varphi:\Loc\to\LCT$ [in the case
$\Phys=\Alg$] let us write $\Sigma(\Mb)$ to denote the equivalence class of the Cauchy surface
of $\Mb$ modulo orientation-preserving diffeomorphisms
and suppose that $\lambda:\Loc\to\NN=\{1,2,\ldots\}$ is constant on such equivalence classes and
obeys $\lambda(\Mb)=1$ if $\Sigma(\Mb)$ is noncompact. Defining $\mu(\Mb) = \max \{\lambda(\Cb): \text{$\Cb$ a component of $\Mb$}\}$, for example,
one can show:
\begin{prop} 
Fix any theory $\Af\in\LCT$ ($\Phys=\Alg$). Then the assignments
\[
\varphi(\Mb) = \Af^{\otimes\mu(\Mb)} \qquad
\varphi(\Mb\stackrel{\psi}{\to}\Nb) = \eta^{\mu(\Mb),\mu(\Nb)}
\]
for all $\Mb\in\Loc$ and hyperbolic embeddings $\psi:\Mb\to\Nb$
define a functor $\varphi\in\Funct(\Loc,\LCT)$, where $\eta^{k,l}$ are the natural 
transformations in Example~\ref{ex:copowers}. If $\Af$ obeys the time-slice axiom, then 
so does $\varphi_\Delta$.
\end{prop}
\begin{proof}
The key point is to show that $\mu$ is monotone with respect to hyperbolic embeddings:
if $\psi:\Mb\to\Nb$ then $\mu(\Mb)\le \mu(\Nb)$. This in turn follows from a result
in Lorentzian geometry which [for connected spacetimes] asserts: if $\psi:\Mb\to\Nb$ is a hyperbolic embedding and $\Mb$ has compact Cauchy surface then $\psi$ is Cauchy and the Cauchy surfaces of 
$\Nb$ are equivalent to those of $\Mb$ under an orientation preserving diffeomorphism.  
The functorial properties of $\varphi$ follow immediately from properties of $\eta^{k,l}$.
If $\Af$ obeys the time-slice axiom, then so do its powers $\Af^{\otimes k}$. Moreover, 
if $\psi$ is Cauchy, then $\Mb$ and $\Nb$ have oriented-diffeomorphic Cauchy surfaces and therefore
$\varphi(\psi)$ is an identity, and $\varphi_\Delta(\psi)=\Af^{\otimes\mu(\Mb)}(\psi)$ is an isomorphism.  
\end{proof}
As a concrete example, which we call the {\em one-field-two-field model}, if we put
\[
\lambda(\Mb) = \begin{cases} 1 & \Sigma(\Mb)~\text{noncompact}
\\ 2 & \text{otherwise} \end{cases}
\]
then the resulting diagonal theory $\varphi_\Delta$ represents one copy of the underlying theory $\Af$ in spacetimes whose Cauchy surfaces are purely noncompact (i.e., have no compact connected components), but two copies in all other spacetimes (see
Fig.~\ref{fig:onefieldtwofield}). 
Moreover, there are obvious subtheory embeddings 
\begin{equation}\label{eq:SPASs_failure}
\Af\nlto\varphi_\Delta\nlto\Af^{\otimes 2}
\end{equation}
which are isomorphisms in some spacetimes but not in
others; the left-hand in spacetimes with purely noncompact Cauchy surfaces but not otherwise, and
\emph{vice versa} for the right-hand embedding. As we have exhibited partial equivalences that are
not equivalences, we conclude that SPASs fails in $\LCT$.

\begin{figure}
\begin{center}
\begin{tikzpicture}[scale=0.8]

\draw[fill=gray,draw=gray] (-7,0.5) -- (-5.5,2) -- (-7,3.5) -- (-8.5,2) -- (-7,0.5);
\node at(-7,-1) {$\varphi_\Delta(\Db)=\Af(\Db)$};
\node at(0,-1)  {$\varphi_\Delta(\Cb)=\Af^{\otimes 2}(\Cb)$};
\node at(-3.5,-1.5) {$\varphi_\Delta(\Cb\sqcup \Db)=\Af^{\otimes 2}(\Cb\sqcup\Db)$};
% Set up axes
\pgfsetxvec{\pgfpoint{1cm}{0}}
\pgfsetzvec{\pgfpoint{0}{1cm}}
\pgfsetyvec{\pgfpoint{0}{0.2cm}}

% Rear portion of the cylinder
\pgfsetstrokecolor{gray}
\pgfsetfillcolor{gray}
\foreach \theta in {0,10,...,170}
{
\pgfpathmoveto{\pgfpointcylindrical{\theta}{2}{0}}
\pgfpathlineto{\pgfpointcylindrical{\theta+10}{2}{0}}
\pgfpathlineto{\pgfpointcylindrical{\theta+10}{2}{4}}
\pgfpathlineto{\pgfpointcylindrical{\theta}{2}{4}}
\pgfpathclose
}
\pgfusepath{fill,stroke}

% Front portion of the cylinder
\pgfsetfillcolor{lightgray}
\pgfsetstrokecolor{lightgray}
\foreach \theta in {180,190,...,350}
{
\pgfpathmoveto{\pgfpointcylindrical{\theta}{2}{0}}
\pgfpathlineto{\pgfpointcylindrical{\theta+10}{2}{0}}
\pgfpathlineto{\pgfpointcylindrical{\theta+10}{2}{4}}
\pgfpathlineto{\pgfpointcylindrical{\theta}{2}{4}}
\pgfpathclose
}
\pgfusepath{fill,stroke}

%% Embedded diamond (upper half)
%\pgfsetfillcolor{blue}
%\pgfsetstrokecolor{blue}
%\pgfpathmoveto{\pgfpointcylindrical{225}{2}{2}}
%\foreach \theta in {225,230,...,315}
%{
%\pgfpathlineto{\pgfpointcylindrical{\theta}{2}{3.5-abs(270-\theta)/30}}
%}
%% Embedded diamond (lower half)
%\pgfpathmoveto{\pgfpointcylindrical{315}{2}{2}}
%\foreach \theta in {315,310,...,225}
%{
%\pgfpathlineto{\pgfpointcylindrical{\theta}{2}{0.5+abs(270-\theta)/30}}
%}
%\pgfclosepath
%\pgfusepath{fill,stroke}

% Bands of latitude
\pgfsetfillcolor{black}
\pgfsetstrokecolor{black}
\foreach \z in {0.5,2,3.5}
{
  \pgfpathmoveto{\pgfpointcylindrical{180}{2}{\z}}
  \foreach \theta in {190,200,...,360}
  {
    \pgfpathlineto{\pgfpointcylindrical{\theta}{2}{\z}}
  }
  \pgfusepath{stroke}
}

\end{tikzpicture}
\end{center}
\caption{The one-field-two-field model assigns one copy of the theory $\Af$ to
a diamond spacetime $\Db$ with noncompact Cauchy surface, but two copies to
a spacetime $\Cb$ with compact Cauchy surface, or spacetimes such as the disjoint union
$\Cb\sqcup\Db$ which have at least one component with compact Cauchy surface.}
\label{fig:onefieldtwofield}
\end{figure}
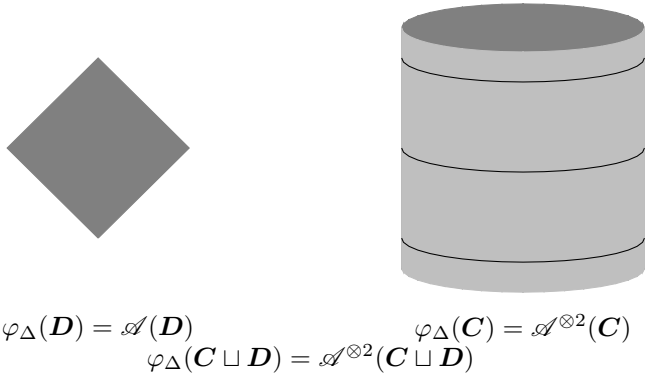
Meditating on this example, we can begin to see the root of the problem. 
If $O$ is any nonempty open globally hyperbolic subset of any $\Mb$ such that $O$ has noncompact 
Cauchy surfaces then the kinematic local algebras obey
\begin{equation}\label{eq:diagonal_kinematic}
\varphi_\Delta^\kin(\Mb;O) = \varphi_\Delta^\kin(\Mb|_O)=\Af
^\kin(\Mb;O),
\end{equation}
so the kinematic local algebras are insensitive to the ambient algebra. 
In one-field-two-field model, for example,
the kinematical algebras of relatively compact regions with nontrivial causal complement
are always embeddings of the corresponding algebra for a single field, even when the
ambient algebra corresponds to two independent fields. This leads to
a surprising diagnosis:  the framework described so far is insufficiently local,
and it is necessary to find another way of describing the local physics. 

The problem, therefore, is to detect the existence of the ambient local degrees of freedom that
are ignored in the kinematical local algebras. The solution is simple: 
where there are physical degrees of freedom, there ought to be energy. In other
words, we should turn our attention to dynamics. 

Some immediate support for the idea that dynamics has a role to play can be
obtained by computing the relative Cauchy evolution of a diagonal theory.
The result is
\begin{equation}\label{eq:rce_calc}
\rce_\Mb^{(\varphi_\Delta)}[\hb] = (\rce_\Mb^{(\varphi)}[\hb])_\Mb\circ \rce_\Mb^{(\varphi(\Mb))}[\hb],
\end{equation}
where $\rce_\Mb^{(\varphi)}[\hb]\in\Aut(\varphi(\Mb))$ is the relative Cauchy evolution of $\varphi$ in
$\LCT$. In known examples this is trivial, but it would be intriguing to find examples in which the choice of  
theory as a function of the spacetime contributes nontrivially to the total stress-energy tensor of the theory.
Provided that $\rce^{(\varphi)}$ is indeed trivial and stress-energy tensors exist, it follows that
\[
[\Ts_\Mb^{(\varphi_\Delta)}(\fb),A] = [\Ts_\Mb^{(\varphi(\Mb))}(\fb),A]
\]
for all $A\in\varphi_\Delta(\Mb)$, 
which means that the stress-energy tensor correctly detects the ambient degrees of 
freedom that are missed in the kinematic description.  

\section{Local physics from dynamics}

Let $\Af$ be a locally covariant theory obeying the time-slice axiom and $O$ be a region in spacetime $\Mb$. The {\em kinematical} description of the physical content of $O$ was given in terms of the physics assigned to $O$ when considered as a spacetime in its own right, i.e., $\Af^\kin(\Mb;O)=\Af(\Mb|_O)$. For the {\em dynamical} description, we focus on that portion of the physical content on $\Mb$ that is unaffected by the geometry in the causal complement of $O$. This can be isolated using the relative Cauchy evolution, which precisely encapsulates the response of the system under a change in the geometry.

Accordingly, for any compact set $K\subset\Mb$,  we define
$\Af^\bullet(\Mb;K)$ to be the maximal subobject of $\Af(\Mb)$
invariant under $\rce_\Mb[\hb]$ for all $\hb\in H(\Mb;K^\perp)$,
where $H(\Mb;K^\perp)$ is the set of metric perturbations in $H(\Mb)$
supported in the causal complement $K^\perp =\Mb\setminus J_\Mb(K)$  of $K$.
Subobjects of this type will exist provided the category $\Phys$ has
equalizers (at least for pairs of automorphisms) and arbitrary set-indexed intersections.
If $\Phys=\Alg$, of course, $\Af^\bullet(\Mb;K)$ is simply the invariant subalgebra. 

From the categorical perspective, it is usually not objects that
are of interest but rather the morphisms between them. Strictly speaking, a subobject of
an object $\Ac$ in a general category is an equivalence class of monomorphisms with codomain $\Ac$,
where $\alpha$ and $\alpha'$ are regarded as equivalent if there is a (necessarily unique)
isomorphism $\beta$ such that $\alpha'=\alpha\circ\beta$. In our case, we are really
interested in a morphism $\alpha^\bullet_{\Mb;K}$ with codomain $\Af(\Mb)$, characterized (up to isomorphism) by the requirements that (a)
\[
\rce_\Mb[\hb]\circ \alpha^\bullet_{\Mb;K} = \alpha^\bullet_{\Mb;K}\qquad \forall \hb\in H(\Mb;K^\perp)
\]
and (b) if $\beta$ is any other morphism with this property then there exists a unique morphism
$\hat{\beta}$ such that $\beta = \alpha^\bullet_{\Mb;K}\circ\hat{\beta}$. The notation
$\Af^\bullet(\Mb;K)$ denotes the domain of $\alpha^\bullet_{\Mb;K}$. For simplicity 
and familiarity, however, we will ignore this issue and write everything in terms of
the objects, rather than the morphisms. This is completely legitimate in categories such as
$\Alg$ or $\CAlg$, in which the objects are sets with additional structure and $\Af^\bullet(\Mb;K)$ can be concretely
constructed as the subset of $\Af(\Mb)$ left fixed by the appropriate relative Cauchy evolution automorphisms. (At the level of detail, however, it is not simply a matter of precision to work with the morphisms: universal definitions, such as that given above for $\alpha^\bullet_{\Mb;K}$, permit
efficient proofs that are `portable' between different choices of the category $\Phys$.)

The subobjects $\Af^\bullet(\Mb;K)$ have many features reminiscent of a net of local algebras in local
quantum physics: for example, isotony ($K_1\subset K_2$ implies $\Af^\bullet(\Mb;K_1)\subset
\Af^\bullet(\Mb;K_2)$) together with 
\[
\Af^\bullet(\Mb;K) = \Af^\bullet(\Mb;K^{\perp\perp}) 
\]
if $K^{\perp\perp}$ is also compact, and 
\begin{align*}
\Af^\bullet(\Mb;K_1)\vee \Af^\bullet(\Mb;K_2) &\subset \Af^\bullet(\Mb; K_1\cup
K_2) \\
\Af^\bullet(\Mb; K_1\cap K_2) &\subset \Af^\bullet(\Mb;K_1)\cap \Af^\bullet(\Mb;K_2) \\
\Af^\bullet(\Mb;\emptyset) &\subset \Af^\bullet(\Mb;K)
\end{align*}
for any compact $K,K_1,K_2$. A rather striking absence is that there is no local covariance property:
it is not true in general that a morphism $\psi:\Mb\to\Nb$ induces isomorphisms between $\Af^\bullet(\Mb;K)$ and $\Af^\bullet(\Nb;\psi(K))$ so as to make the diagram
\begin{equation}
\begin{tikzpicture}[baseline=0 em,description/.style={fill=white,inner sep=2pt}]
\matrix (m) [ampersand replacement=\&,matrix of math nodes, row sep=3em,
column sep=2.5em, text height=1.5ex, text depth=0.25ex]
{ \Af^\bullet(\Mb;K) \&  \Af^\bullet(\Nb;\psi(K)) \\
 \Af(\Mb) \&  \Af(\Nb)\\ };
\path[->,font=\scriptsize]
(m-1-1) edge node[left] {$\alpha^\bullet_{\Mb;K}$} (m-2-1)
(m-1-2) edge node[auto] {$\alpha^\bullet_{\Nb;\psi(K)}$} (m-2-2)
(m-2-1) edge node[auto] {$\Af(\psi)$} (m-2-2);
%(m-2-2) edge node[below] {$ \varphi(\psi)_\Nb $} (m-2-3)
%(m-1-2) edge[color=red,ultra thick] node[above,sloped,color=red] {$ \varphi_\Delta(\psi) $} (m-2-3);
\path[->,dotted,font=\scriptsize]
(m-1-1) edge node[above]  {$\cong$} node[below] {?} (m-1-2);
\end{tikzpicture}\label{eq:bullet_covariance_diagram}
\end{equation}
commute (unless $\psi$ is itself an isomorphism). 
Again, this can be seen from the 
example of diagonal theories, for which we have
\[
\varphi_\Delta^\bullet(\Mb;K)=\varphi(\Mb)^\bullet(\Mb;K)
\]
as a consequence of \eqref{eq:rce_calc}. In the one-field-two-field model, for example,
a hyperbolic embedding of $\psi:\Db\to\Cb$ (where $\Db$ and $\Cb$ are as in Fig.~\ref{fig:onefieldtwofield}) we have $\varphi_\Delta^\bullet(\Db;K) = \Af^\bullet(\Db;K)$
but $\varphi_\Delta^\bullet(\Cb;\psi(K)) = \Af^{(\otimes 2)\bullet}(\Cb;\psi(K))$, for any compact
subset $K$ of $\Db$. In general, these subobjects will not be isomorphic in the required sense. This situation should be compared with the local covariance property enjoyed by the kinematic description, expressed by
the commuting diagram \eqref{eq:kinematic_covariance_diagram}.

The discussion so far has shown how dynamics allows us to associate a subobject $\Af^\bullet(\Mb;K)$ of $\Af(\Mb)$ with every compact subset $K$ of $\Mb$. As mentioned above, for simplicity we are taking a concrete viewpoint in which $\Af^\bullet(\Mb;K)$ is a subset of $\Af(\Mb)$. For purposes of comparison
with the kinematic viewpoint, we must also associate subobjects of $\Af(\Mb)$ with open subsets of $\Mb$. 
Given an open subset $O$ of $\Mb$, the basic idea will be to take the subobject generated by all
$\Af^\bullet(\Mb;K)$ indexed over a suitable class of compact subsets $K$ of $\Mb$. To define the class
of $K$ involved, we must introduce some terminology. 

Following~\cite{BrunettiRuzzi_topsect}, a {\em diamond} is defined to be an open relatively compact subset of $\Mb$, taking the form $D_\Mb(B)$ where $B$, the {\em base} of the diamond, is a subset of a Cauchy surface $\Sigma$ so that 
in a suitable chart $(U,\phi)$ of $\Sigma$, $\phi(B)$ is a nonempty open ball in $\RR^{n-1}$ whose closure is contained in $\phi(U)$. (A given diamond has many possible bases.) By a {\em multi-diamond}, we understand any finite union of mutually causally disjoint diamonds; the base of a multi-diamond is formed by any union of bases for each component diamond. 

Given these definitions, for any open set $O$ in $\Mb$, let $\KK(\Mb;O)$ be the set of 
compact subsets of $O$ that have a multi-diamond neighbourhood with a base contained in $O$. We then define
\[
\Af^{\rm dyn}(\Mb;O) = \bigvee_{K\in \KK(\Mb;O)} \Af^\bullet(\Mb;K)
\]
that is, the smallest $\Phys$-subobject of $\Af(\Mb)$ that contains all the $\Af^\bullet(\Mb;K)$ indexed
by $K\in \KK(\Mb;O)$. Of course, the nature of the category $\Phys$ enters here: if it is $\Alg$ (resp., 
$\CAlg$)
then $\Af^{\rm dyn}(\Mb;O)$ is the ($C$)${}^*$-subalgebra of $\Af(\Mb)$ generated by the $\Af^\bullet(\Mb;K)$; if $\Phys$ is $\Sympl$, then the linear span is formed. More abstractly,
we are employing the categorical join of subobjects, which can be given a universal definition in terms of the 
morphisms $\alpha^\bullet_{\Mb;K}$, and gives a morphism (characterized up to isomorphism)
$\alpha^\dyn_{\Mb;O}$ with codomain $\Af(\Mb)$; $\Af^\dyn(\Mb;O)$ is the domain
of some such morphism. In $\Alg$ (and more generally, if $\Phys$ has
pullbacks, and class-indexed intersections) the join has the stronger property of being
a categorical union~\cite{DikranjanTholen}, and this is used in our general setup.

It is instructive to consider these algebras for diagonal theories $\varphi_\Delta$, 
under the assumption that $\rce^{(\varphi)}$ is trivial. Using \eqref{eq:rce_calc}, it
follows almost immediately that
\[
\varphi_\Delta^\bullet(\Mb;K)=\varphi(\Mb)^\bullet(\Mb;K), \qquad
\varphi_\Delta^\dyn(\Mb;O)=\varphi(\Mb)^\dyn(\Mb;O)
\]
for all compact $K\subset \Mb$ and open $O\subset \Mb$. This should
be contrasted with the corresponding calculations for kinematic algebras
in \eqref{eq:diagonal_kinematic}; we see that the dynamical algebras
sense the ambient spacetime in a way that the kinematic algebras do not.
In the light of this example, the following definition is very natural.

\begin{definition} $\Af$ obeys \emph{dynamical locality} if
\[
\Af^{\rm dyn}(\Mb;O) \cong \Af^{\rm kin}(\Mb;O) 
\]
for all nonempty open globally hyperbolic subsets $O$ of $\Mb$ with finitely many components,
where the isomorphism should be understood as asserting the existence of isomorphisms
$\beta_{\Mb;O}$ such that $\alpha^\dyn_{\Mb;O}= \alpha^\kin_{\Mb;O}\circ
\beta_{\Mb;O}$; i.e., equivalence as subobjects of $\Af(\Mb)$. 
\end{definition}

In the next section, we will see that the class of dynamically local theories in $\LCT$ 
has the SPASs property, our main focus in this contribution. However, we note that
dynamical locality has a number of other appealing consequences, which suggest that it may
prove to be a fruitful assumption in other contexts. 

First, any dynamically local theory $\Af$ is necessarily additive, in the sense that
\[
\Af(\Mb) = \bigvee_{O} \Af^\dyn(\Mb;O),
\]
where the categorical union is taken over a sufficiently large class of 
open globally hyperbolic sets $O$; for example, the truncated multi-diamonds 
(intersections of a multi-diamond with a
globally hyperbolic neighbourhood of a Cauchy surface containing its base).

Second, the `net' $K\mapsto \Af^\bullet(\Mb;K)$ is locally covariant: if
$\psi:\Mb\to\Nb$ and $K\subset \Mb$ is {\em outer regular} then there is
a unique isomorphism making the diagram~\eqref{eq:bullet_covariance_diagram} commute.
Here, a compact set $K$ is said to be outer regular if there exist relatively compact nonempty
open globally hyperbolic subsets $O_n$ ($n\in\NN$) with finitely many components such that 
(a) $\operatorname{cl}(O_{n+1})\subset O_n$ and $K\in\KK(\Mb;O_n)$ for each $n$ and (b) $K=\bigcap_{n\in\NN} O_n$.  

Third, in the case $\Phys=\Alg$ or $\CAlg$,\footnote{One may also formulate an analogous statement for more general categories $\Phys$.} if $\Af^\bullet(\Mb;\emptyset)=\CC\II_{\Af(\Mb)}$ and $\Af$ is dynamically local then the theory obeys {\em extended locality}~\cite{Landau1969}, i.e.,  if $O_1, O_2$ are causally disjoint nonempty globally hyperbolic subsets then $\Af^\kin(\Mb;O_1)\cap \Af^\kin(\Mb;O_2)=\CC \II_{\Af(\Mb)}$. Conversely, we can infer triviality of $\Af^\bullet(\Mb;\emptyset)$ from extended locality.

Fourth, there is an interesting application to the old question of whether
there can be any preferred state in general spacetimes. Suppose that
$\Af$ is a theory in $\LCT$ with $\Phys$ taken to be $\Alg$ or $\CAlg$. 
A {\em natural state} $\omega$ of the theory is an assignment $\Mb\mapsto\omega_\Mb$
of states $\omega_\Mb$ on $\Af(\Mb)$, subject to the contravariance
requirement that $\omega_\Mb=\omega_\Nb\circ\Af(\psi)$ for
all $\psi:\Mb\to\Nb$. Hollands and Wald~\cite{Ho&Wa01} give an argument
to show that this is not possible for the free scalar field; likewise, BFV also sketch an argument to this effect 
(also phrased concretely in terms of the free field). 
The following brings these arguments to a sharper and model-independent form:
\begin{thm}\label{thm:natural_states}
Suppose $\Af$ is a dynamically local theory in $\LCT$ equipped with a natural state
$\omega$. If there is a spacetime $\Mb$ with noncompact Cauchy surfaces such that
$\omega_\Mb$ induces a faithful representation of $\Af(\Mb)$ with the Reeh--Schlieder property,
then the relative Cauchy evolution in $\Mb$ is trivial. Moreover, if extended locality holds in $\Mb$
then $\Af$ is equivalent to the trivial theory. 
\end{thm}
By the Reeh--Schlieder property, we mean that the GNS vector $\Omega_\Mb$
corresponding to $\omega_\Mb$ is cyclic for the induced representation of
each $\Af^\kin(\Mb;O)$ for any open, globally hyperbolic, relatively compact
$O\subset \Mb$. This requirement will be satisfied if, for example, 
the theory in Minkowski space obeys standard assumptions of AQFT with the
natural state reducing to the Minkowski vacuum state.
The trivial theory $\If$ is the theory assigning the trivial unital algebra to each $\Mb$ 
and its identity morphism to any morphism in $\LCT$.
\begin{proof}
(Sketch) We first argue that $\omega_\Mb\circ\rce_\Mb[\hb] =\omega_\Mb$ for all 
$\hb\in H(\Mb)$ and all $\Mb\in\Loc$ by virtue of the natural state assumption and the definition of the
relative Cauchy evolution. In the GNS representation $\pi_\Mb$ of $\omega_\Mb$, 
the relative Cauchy evolution can be unitarily represented by unitaries leaving
the vacuum vector $\Omega_\Mb$ invariant. Choose any $\hb\in H(\Mb)$ and
an open, relatively compact, globally hyperbolic $O$ spacelike separated from $\supp \hb$.
By general properties of the relative Cauchy evolution, we have 
$\rce_\Mb[\hb]\circ\alpha^\kin_{\Mb;O}=
\alpha^\kin_{\Mb;O}$ and hence that the unitary $U_\Mb[\hb]$ implementing
$\rce_\Mb[\hb]$ is equal to the identity on the subspace $\pi_\Mb(\Af(\iota_{\Mb;O})A)\Omega_\Mb$
($A\in\Af(\Mb|_O)$) which is dense by the Reeh--Schlieder property. As $\pi_\Mb$
is faithful, the relative Cauchy evolution is trivial as claimed. 

It follows that $\Af^\bullet(\Mb;K)=\Af(\Mb)$ for all compact $K$; consequently, by
dynamical locality, $\Af^\kin(\Mb;O)=\Af^\dyn(\Mb;O)=\Af(\Mb)$ for all nonempty
open, globally hyperbolic $O$. Taking two such $O$ at spacelike separation and 
applying extended locality, we conclude that $\Af(\Mb)=\CC\II$. 
The remainder of the proof will be
given in the next section.
\end{proof}

\section{SPASs at last}

Our aim is to show that the category of dynamically local theories has the SPASs property.

\begin{thm}\label{thm:SPASs}
Suppose $\Af$ and $\Bf$ are theories in $\LCT$ obeying dynamical locality.
Then any partial equivalence $\zeta:\Af\nto\Bf$ is an equivalence. 
\end{thm}
\begin{proof} 
(Sketch.) The main ideas are as follows. First, we show that if $\zeta_\Nb$
is an isomorphism then $\zeta_\Lb$ is an isomorphism for every spacetime
$\Lb$ for which there is a morphism $\psi:\Lb\to\Nb$. Second, if
$\psi:\Lb\to\Nb$ is Cauchy, then $\zeta_\Lb$ is an isomorphism if and only if
$\zeta_\Nb$ is. 

As $\zeta$ is a partial equivalence, there exists a spacetime $\Mb$ for
which $\zeta_\Mb$ is an isomorphism. It follows that $\zeta_\Db$ is an isomorphism
for every multi-diamond spacetime $\Db$ that can be embedded in $\Mb$.
Using deformation arguments \cite{FullingNarcowichWald} there is a chain of 
Cauchy morphisms linking each such multi-diamond spacetime to any other
multi-diamond spacetime with the same number of connected components. Hence
$\zeta_\Db$ is an isomorphism for all multi-diamond spacetimes $\Db$.
Now consider any other spacetime $\Nb$. Owing to dynamical locality,
both $\Af(\Nb)$ and $\Bf(\Nb)$ are generated over subobjects that
correspond to diamond spacetimes, which (it transpires)
are isomorphic under the restriction of $\zeta_\Nb$. By the properties
of the categorical union, it follows that $\zeta_\Nb$ is an isomorphism.
\end{proof}
We remark that the
chain of partial equivalences in~\eqref{eq:SPASs_failure} shows that
we cannot relax the condition that {\em both} $\Af$ and $\Bf$ be dynamically local.

The foregoing result allows us to conclude the discussion of natural states:
\begin{proof}[End of the proof of Theorem~\ref{thm:natural_states}]
Note that there is a natural transformation
$\zeta:\If\nto\Af$, such that each $\zeta_\Nb$ simply embeds the trivial unital algebra in
$\Af(\Nb)$. For the spacetime $\Mb$ given in the hypotheses, we already know 
that $\Af(\Mb)=\CC\II$, so $\zeta_\Mb$ is an isomorphism. Hence
by dynamical locality and Theorem~\ref{thm:SPASs}, $\zeta$ is an equivalence
so $\Af\cong\If$. 
\end{proof}

Dynamical locality also significantly reduces our freedom to 
construct pathological theories.
\begin{thm} If $\varphi_\Delta$ is a diagonal theory (with $\rce^\varphi$
trivial) such that $\varphi_\Delta$ and every $\varphi(\Mb)$ are
dynamically local, then all $\varphi(\Mb)$ are equivalent. If the
$\varphi(\Mb)$ have trivial automorphism group then $\varphi_\Delta$ is equivalent 
to each of them.
\end{thm}
\begin{proof} (Sketch) Consider any morphism $\psi:\Mb\to\Nb$ in $\Loc$. 
Using dynamical locality of $\varphi_\Delta$ and $\varphi(\Nb)$, we may deduce that $\varphi(\psi)_\Mb$ is an isomorphism and hence $\varphi(\psi)$ is an equivalence. As any two spacetimes in $\Loc$
can be connected by a chain of (not necessarily composable) morphisms, it follows that 
all the $\varphi(\Mb)$ are equivalent. 

In particular, writing $\Mb_0$ for Minkowski space, the previous argument allows us to choose
 an equivalence $\zeta_{(\Mb)}:\varphi(\Mb_0)\nto \varphi(\Mb)$ for each spacetime $\Mb$ (no uniqueness is assumed). Every morphism $\psi:\Mb\to\Nb$ then induces an automorphism
$\eta(\psi) = \zeta_{(\Mb)}^{-1}\circ\varphi(\psi)\circ\zeta_{(\Nb)}$ of $\varphi(\Mb_0)$. 
As $\Aut(\varphi(\Mb_0))$ is assumed trivial, it follows that $\varphi(\psi) = \zeta_{(\Mb)}\circ\zeta_{(\Nb)}^{-1}$ 
and it is easy to deduce that there is an equivalence $\zeta:\varphi(\Mb_0)\nto\varphi_\Delta$
with components $(\zeta_{(\Mb)})_\Mb$. 
\end{proof}
We remark that the automorphism group of a theory in $\LCT$ can be interpreted as its group of 
global gauge invariances~\cite{Fewster_in_prep} and that a theory in which the $\Af(\Mb)$
are algebras of observables will have trivial gauge group. The argument just given has a 
cohomological flavour, which can also be made precise and indicates that a cohomological study of
$\Loc$ is worthy of further study. 

%\vspace{0.2cm}
%Idea:
%\begin{itemize}
%\item Deduce $\zeta_{\Mb|_O}$ is an isomorphism for some diamond.
%\item Deduce that $\zeta_\Db$ is an isomorphism for all diamond spacetimes
%\item Given any $\Nb$, use a generating set of diamonds to infer that
%$\zeta_\Nb$ is an isomorphism.
%\end{itemize}
%\end{frame}
%
%\begin{frame}{Strongly local diagonal theories}
%Suppose $\varphi_\Delta$ is a diagonal theory (with $\rce^\varphi$
%trivial) such that $\varphi_\Delta$ and every $\varphi(\Mb)$ are
%strongly local. 
%
%\vspace{0.5cm}
%Then
%\begin{itemize}
%%\item whenever $\Mb\stackrel{\psi}{\to}\Nb$, $\varphi(\psi):\varphi(\Mb)\nto\varphi(\Nb)$ is an
%%equivalence 
%\item $\varphi(\Mb)\cong\varphi(\Nb)$ for all $\Mb,\Nb\in\Man$
%\item %$\varphi_\Delta$ is \alert{gauge-equivalent} to any $\varphi(\Mb_0)$, i.e., 
%$\varphi_\Delta\cong \Af$ for $\Af$ with
%\[
%\Af(\psi) = {\color{red}\eta(\psi)_\Nb}\circ\varphi(\Mb_0)(\psi)
%\]
%for every $\psi:\Mb\to\Nb$, where
%$\eta\in\Funct(\Man,\Aut(\varphi(\Mb_0)))$ and $\Mb_0\in\Man$ is arbitrary.
%\item If $\Aut(\varphi(\Mb_0))$ is trivial, then $\varphi_\Delta\cong\varphi(\Mb_0)$. 
%\end{itemize}
% 
%\end{frame}
%

\section{Example: Klein--Gordon theory}\label{sect:KG}

The abstract considerations of previous sections would be of questionable relevance if
they were not satisfied in concrete models. We consider the standard example of the
free scalar field, discussed above.

First let us consider the classical theory. Let $\Sympl$ be the category of
real symplectic spaces with symplectomorphisms as the morphisms. To each $\Mb$, we assign the space $\Lf(\Mb)$ of smooth, real-valued solutions
to $(\Box+m^2)\phi=0$, with compact support on 
Cauchy surfaces; we endow $\Lf(\Mb)$ with the standard (weakly nondegenerate) symplectic product
\[
\sigma_\Mb(\phi,\phi') = \int_\Sigma (\phi n^a\nabla_a\phi' - \phi' n^a\nabla_a\phi) 
d\Sigma
\]
for any Cauchy surface $\Sigma$. 
To each hyperbolic embedding $\psi:\Mb\to\Nb$ there is a symplectic map 
$\Lf(\psi):\Lf(\Mb)\to\Lf(\Nb)$ so that $E_\Nb\psi_*f = \Lf(\psi)E_\Mb f$
for all $f\in\CoinX{\Mb}$ (see BFV or~\cite{BarGinouxPfaffle}).

The relative Cauchy evolution for $\Lf$ was computed in BFV (their Eq.~(15))
as was its functional derivative with respect to metric perturbations (see pp.~61-62 of BFV).
An important observation is that this can be put into a nicer form: 
writing
\[
F_\Mb[f]\phi = \left.\frac{d}{ds}\rce^{(\Lf)}_\Mb[f]\phi \right|_{s=0}
\]
it turns out that
\begin{equation}\label{eq:magic_formula}
\sigma_\Mb(F_\Mb[f]\phi,\phi) = \int_\Mb f_{ab} T^{ab}[\phi]\,\dvol_\Mb,
\end{equation}
where 
\[
T_{ab}[\phi] = \nabla_a\phi\nabla_b\phi -\frac{1}{2}g_{ab}\nabla^c\phi\nabla_c\phi +\frac{1}{2}m^2 g_{ab}\phi^2
\] 
is the classical stress-energy tensor of the solution $\phi$.
(In passing, we note that \eqref{eq:magic_formula} provides an explanation, for the free scalar field and similar theories, as to why the relative Cauchy evolution can be regarded as equivalent to the specification of the classical action.)

We may use these results to compute $\Lf^\bullet(\Mb;K)$ and $\Lf^\dyn(\Mb;O)$:
they consist of solutions whose stress-energy tensors vanish
in $K^\perp$ (resp., $O'=\Mb\setminus\operatorname{cl} J_\Mb(O)$).  
For nonzero mass, the solution must vanish wherever the stress-energy does and
we may conclude that $\Lf^\dyn(\Mb;O)=\Lf^\kin(\Mb;O)$, i.e., we have
dynamical locality of the classical theory $\Lf$. 

At zero mass, however, we can deduce only that the solution is constant in regions where its stress-energy tensor vanishes, so
$\Lf^\bullet(\Mb;K)$ consists of solutions in $\Lf(\Mb)$ that are
constant on each connected component of $K^\perp$. In 
particular, if $K^\perp$ is connected and $\Mb$ has noncompact
Cauchy surfaces this forces the solutions to vanish in $K^\perp$
as in the massive case. However, if $\Mb$ has compact Cauchy 
surfaces this argument does not apply and indeed the constant
solution $\phi\equiv 1$ belongs to
every $\Lf^\bullet(\Mb;K)$ and $\Lf^\dyn(\Mb;O)$, but it does not belong to $\Lf^\kin(\Mb;O)$
unless $O$ contains a Cauchy surface of $\Mb$. Hence the massless theory fails to be dynamically local. 
The source of this problem is easily identified: it arises from the
global gauge symmetry $\phi\mapsto\phi+\text{const}$ in the classical action.

At the level of the quantized theory, one may show that similar results hold: the $m>0$
theory is dynamically local, while the $m=0$ theory is not. We argue that this should
be taken seriously as indicating a (fairly mild) pathology of the massless minimally coupled scalar field,
rather than a limitation of dynamical locality. In support of this position we note:
\begin{itemize}
\item Taking the gauge symmetry seriously, we can alternatively quantize the theory of
currents $j=d\phi$; this turns out to be a well-defined locally covariant and dynamically local theory
in dimensions $n>2$. While dynamical locality fails for this model in $n=2$ dimensions in the present setting, it may be restored by restricting the scope of the theory to connected spacetimes. 
\item The constant solution is also the source of another well-known problem: there is
no ground state for the theory in ultrastatic spacetimes with compact spatial section
(see, for example~\cite{FullRuij87}). The
same problem afflicts the massless scalar field in two-dimensional Minkowski space, where it
is commonplace to reject the algebra of fields in favour of the algebra of currents.
\item The nonminimally coupled scalar field (which does not have the gauge symmetry)
is dynamically local even at zero mass (a result due to Ferguson~\cite{Ferg_in_prep}).
\end{itemize}
Actually, this symmetry has other interesting aspects: it is spontaneously broken in 
Minkowski space~\cite{Streater_broken}, for example, and the automorphism
group of the functor $\Af$ is noncompact: $\Aut(\Af) =\ZZ_2\ltimes \RR$~\cite{Fewster_in_prep}.

\section{Summary and outlook}

We have shown that the notion of the `same physics in all spacetimes' can be given a formal meaning (at least in part)
and can be analysed in the context of the BFV framework of locally covariant theories. While
local covariance in itself does not guarantee the SPASs property, the dynamically local theories
do form a class of theories with SPASs; moreover, dynamical locality seems to be a natural
and useful property in other contexts. Relative Cauchy evolution enters the discussion in an
essential way, and seems to be the replacement of the classical action in the axiomatic setting. A key question, given our starting point,
is the extent to which the SPASs condition is sufficient as well as necessary for a class of theories to represent the same physics in all spacetimes. As a class of theories that had no subtheory embeddings
other than equivalences would satisfy SPASs, there is clearly scope
for further work on this issue. In particular, is it possible to 
formulate a notion of `the same physics on all spacetimes' in terms
of individual theories rather than classes of theories? 

In closing, we remark that the categorical framework opens a completely new way of analysing quantum
field theories, namely at the functorial level. It is conceivable that all structural properties of QFT should
have a formulation at this level, with the instantiations of the theory in particular spacetimes
taking a secondary place. Lest this be seen as a flight to abstraction, we emphasize again that
this framework is currently leading to new viewpoints and concrete calculations in cosmology and elsewhere. This provides all the more reason to understand why theories based on our experience with
terrestrial particle physics can be used in very different spacetime environments while preserving
the same physical content.

\end{document}